\documentclass{artjlt}

\usepackage{amsmath,amsfonts,amssymb}

\title{Orbital reducibility and a  generalization\\
of lambda symmetries}                                     
\author{Giampaolo Cicogna\\
Giuseppe Gaeta\\
Sebastian Walcher}                 
\lastname{Cicogna, Gaeta, Walcher}  



\address{G. Cicogna\\
      Dipartimento di Fisica\\
 Universit\`a di Pisa\\
        and INFN, Sezione di Pisa\\
       Largo B.~Pontecorvo 3\\
     I-56127 Pisa, Italy\\
cicogna@df.unipi.it\\
  \\
G. Gaeta\\
 Dipartimento di Matematica\\
Universit\`a degli Studi di Milano\\
via Saldini 50\\
I-20133 Milano, Italy\\
giuseppe.gaeta@unimi.it\\
  \\
S. Walcher\\
Lehrstuhl A f\"ur Mathematik\\
 RWTH Aachen\\
D-52056 Aachen, Germany\\
walcher@matha.rwth-aachen.de
}

\begin{document}


\maketitle

\begin{abstract}
We review the notion of reducibility and we introduce and discuss the notion of orbital reducibility for autonomous ordinary differential equations of first order. The relation between (orbital) reducibility and (orbital) symmetry is investigated and employed to construct (orbitally) reducible systems. By standard identifications, the notions extend to non-autonomous ODEs of first and higher order. Moreover we thus obtain a generalization of the lambda symmetries of Muriel and Romero. Several examples are given.
\end{abstract}

\section{Introduction and preliminaries}
In the present paper we discuss reducibility and reduction for ordinary differential equations. Our focus of interest is on (explicit) reducibility imparted by some map to a lower-dimensional system (possibly defined on a submanifold of $n$-space).
For equations of higher order, reduction is frequently understood as reduction of order, which will also be considered.
The main purpose of the paper is to consider various notions of reducibility, establish a general framework, and elucidate the relations between symmetry, generalizations such as $\lambda$-symmetry, and reducibility. We clarify and extend notions, generalize results, and obtain new applications. In a related paper \cite{CGWsprol} we present a more thorough discussion of higher order equations, including prolongation formulas.

For autonomous first order differential equations (resp. the associated vector fields) there are well-defined and well-established notions of symmetry (sending parameterized solutions to parameterized solutions) and orbital symmetry (sending solution orbits to solution orbits, and consequently invariant sets to invariant sets). A canonical notion of reducibility (which includes symmetry reduction) was introduced and discussed in  \cite{HaWa}, and we extend this by introducing the notion of orbital reducibility. The latter turns out to correspond to a generalization of $\lambda$-symmetries  (Muriel and Romero; in particular \cite{MuRo1a}). Due to a relation between (orbital) reducibility and (orbital) symmetry there is a canonical construction of (orbitally) reducible systems from (orbitally) symmetric ones, and in some cases it can be shown that all (orbitally) reducible systems are obtained in this manner. For vector fields the notions of "reducibility via some map" discussed in this paper seem to be comprehensive. Moreover, all notions of reducibility for non-autonomous systems or for equations of higher order can be traced back to the case of first order autonomous systems, just as non-autonomous systems or systems of higher order can be written as first-order autonomous systems. 
Reducibility for a non-autonomous first-order system (via a map sending solutions to solutions of some lower dimensional system) amounts to orbital reducibility of an associated autonomous system. This general framework for reduction comprises the main focus of the paper.

Reducibility  of a higher-order equation (in the sense of reducing order)  is equivalent to reducibility of the canonically associated first-order system (in the sense of reducing dimension). In the course of proving this, we note that every (non-autonomous) $m$-dimensional system of order one with nontrivial right-hand side may be rewritten as a single equation of order $m+1$, reverting the usual procedure.

The problem to explicitly determine a reducing (or orbitally reducing) map for a given autonomous equation of first order seems to be just as hard (and as algorithmically inaccessible) as the problem of finding a symmetry (or orbital symmetry). This is essentially due to the straightening theorem and the implicit function theorem, which are not, or not completely, constructive. But the inverse problem to determine all vector fields reducible by a given map is easier to access in some relevant cases.
In particular we transfer the approach from Olver and Rosenau  \cite{OlRo2} to the ordinary differential equation setting, and determine all differential equations which admit reduction by invariants of a given compact and connected group.

Moreover we construct reducible higher-order equations from equations admitting symmetries, in particular Lie point symmetries. In this way we extend the class of reducible equations obtained from lambda-symmetric systems, and provide a different perspective for the latter.

Throughout this paper we restrict attention to analytic functions and vector fields; many of the results can, with some care, be extended to the smooth case. In order to give a self-contained discussion, and to make           the paper accessible to readers with different backgrounds, we include a review (and sometimes rephrasing) of some facts and methods.

\section{Reducibility and orbital reducibility}
We first fix some notation.
Let an (analytic) autonomous ordinary differential equation
\begin{equation}\label{ode}
\dot x = f(x)
\end{equation}
be given on the open and connected subset $U$ of $\mathbb K^n$ (with $\mathbb K$ standing for $\mathbb R$ or $\mathbb C$). 
We denote by $X_f$ the corresponding Lie derivative which acts on analytic functions via
\[
\phi\mapsto X_f(\phi),\quad X_f(\phi)\,(x):=D\phi(x)\,f(x),
\]
and recall that $\phi$ is called a first integral of  \eqref{ode} if $X_f(\phi)=0$.
We distinguish between solutions (including parameterization) and solution orbits (trajectories). If one is primarily interested in orbits, then it is appropriate to consider an equivalence class of differential equations, rather than the single equation \eqref{ode} in case $f\not=0$. Let $Z$ be the zero set of $f$. Then a differential equation defined on the open set  $\widetilde U\subseteq U$ has the same solution orbits on $\widetilde U\setminus Z$ if and only if it has the form
\[
\dot x=\mu(x)\cdot f(x)\quad\mbox{  on }\widetilde U\setminus Z,
\]
with $\mu:\, \widetilde U\setminus Z\to \mathbb K$ analytic and without zeros. (The non-obvious direction holds because a suitable reparameterization of a solution of \eqref{ode} will produce a solution of $\dot x=\mu(x)\cdot f(x)$.) Thus one has the notion of local orbital equivalence for vector fields that are defined on some open and dense subset of $U$. 
\begin{remark}\label{orbfi}{\em 
Two equations are locally orbit-equivalent if and only if they admit the same first integrals near any non-stationary point. Stated in a different way, two equations are locally orbit-equivalent if and only if they admit the same local invariant sets near any non-stationary point.}\quad$\diamond$
\end{remark}
\begin{remark}\label{autonon}{\em 
For a non-autonomous equation
\[
\dot x=q(t,x)\quad\mbox{ on }V\subseteq \mathbb K\times\mathbb K^n
\]
one may define the ''autonomized'' system
\[
\begin{array}{lll}
\dot x_0&=& 1\\
\dot x&=& q(x_0,x)
\end{array}, \quad (x_0,x)\in V.
\]
A reverse to this procedure is obtained as follows: If 
\[
f(x)=\left(\begin{array}{c}
                            f_1(x)\\
                          \vdots\\
                            f_n(x)\end{array}\right)
\]
and one of the components, say $f_1$, has no zero on the open subset $\widehat U$, then passing to an orbit equation
\[
\begin{array}{ccc}
{dx_2}/{dx_1}&=& {f_2(x)}/{f_1(x)}\\
                             &\vdots&                   \\
{dx_n}/{dx_1}&=& {f_n(x)}/{f_1(x)}
\end{array}
\]
provides a non-autonomous system whose autonomization is locally orbitally equivalent to \eqref{ode} on an open and dense subset of $U$. Thus non-autonomous equations in dimension $n$ and (local) orbital equivalence classes of autonomous equations in dimension $n+1$ stand in correspondence.}\quad$\diamond$
\end{remark}
\subsection{Symmetries and orbital symmetries: Review}
A symmetry of the autonomous differential equation \eqref{ode} is a (locally invertible) map sending parameterized solutions to parameterized solutions.  An orbital symmetry of \eqref{ode} is a (locally invertible) transformation mapping solution orbits to solution orbits, hence sending \eqref{ode} to an orbit-equivalent equation $\dot x=\mu(x)\cdot f(x)$. We recall a characterization of infinitesimal (orbital) symmetries; see e.g. Olver \cite{Olv}, Chapter 2, in particular Exercise 2.19, or see \cite{WMul}.
\begin{proposition}\label{symcrit} Let $g$ be a vector field on some open subset of $U$. Then:\\
{\em(a)} The local transformation group generated by $g$ consists of local symmetries of $\dot x=f(x)$ if and only if $\left[g,f\right]=0$.\\
{\em(b)} The local transformation group generated by $g$ consists of local orbital symmetries of $\dot x=f(x)$ if and only if $\left[g,f\right]=\alpha\cdot f$ for some scalar function $\alpha$.
\end{proposition}
For non-autonomous equations
\[
\dot x=q(t,x)\quad\mbox{ on }V\subseteq \mathbb K\times\mathbb K^n
\]
the usual definition of a symmetry is that of a locally invertible map defined on some open subset of $V$ (thus transforming both $t$ and $x$) such that solutions are mapped to solutions, see e.g. Olver \cite{Olv}. Equivalently, by the above Proposition and Remarks \ref{orbfi} and \ref{autonon}, such a map is an orbital symmetry for any autonomized system.
Thus, necessary and sufficient conditions defining an infinitesimal symmetry of a non-autonomous first-order equation are known.

\subsection{Reducibility: Basic notions and results}
The intent underlying any notion of "reducibility by some map" for equation \eqref{ode} is rather obvious, as noted in the Introduction. But details have to be specified. The following particular form was introduced in \cite{HaWa}.
\begin{Definition}\label{redudef}
We call the equation \eqref{ode} {\em reducible on $U$} if there exist a number $m$, $0<m<n$, a positive integer $r$, an analytic map $\Psi:\,U\to \mathbb K^r$, and a differential equation
\[
\dot y=h(y)
\]
defined on an open neighborhood $\widetilde U$ of $\Psi(U)$ such that $\Psi$ maps parameterized solutions of $\dot x=f(x)$ to parameterized solutions of $\dot y=h(y)$, and the derivative $D\Psi(x)$ has rank $\leq m$ on $U$, with rank $=m$ at some point.
\end{Definition}
 The solution-preserving property is equivalent to the identity
\begin{equation}\label{reducond}
D\Psi(x)f(x)=h(\Psi(x))\mbox{    on }U.
\end{equation}
Moreover, due to our assumptions the derivative $D\Psi(x)$ has maximal rank $m$ on an open-dense subset of $U$.

This definition of reducibility is designed to include interesting cases, like reduction by group invariants, which do not a priori provide a map to a vector space of smaller dimension. But locally near any maximal rank point, one has reduction to $\mathbb K^m$, and the structure of reducible vector fields is quite simple. In this sense, the problem can be locally ''trivialized''.

\begin{lemma}\label{loctriv} Let $y\in U$ such that $D\Psi(y)$ has rank $m$, and let $\psi_1,\ldots,\psi_r$ denote the entries of $\Psi$. Then there is a neighborhood $\widetilde U$ of $y$ such that up to a coordinate transformation one may assume that $\psi_1=x_1,\ldots,\psi_m=x_m$, and $\widehat \Psi:=\left(\psi_1,\ldots,\psi_m\right)^{\rm tr}$ is then a reducing map to $\mathbb K^m$. Moreover, up to this coordinate transformation one has
\[
f(x)=\left(\begin{array}{c}f_1(x_1,\ldots, x_m)\\
                                              \vdots\\
                                         f_m(x_1,\ldots, x_m)\\
                                                 \ast \\
                                              \vdots\\
                                                  \ast\end{array}\right)
\]
where the asterisks symbolize functions of all variables $x_1,\ldots, x_n$.
\end{lemma}
\begin{proof} We may assume by the implicit function theorem and the rank condition that $\psi_1=x_1,\ldots,\psi_m=x_m$. Since the matrix $(D\psi_1,\ldots, D\psi_r)^{\rm tr}$ has rank $m$, the functions $\psi_{m+1},\ldots, \psi_r$ depend on $x_1,\ldots,x_m$ only. Then by the reducing property $f_1,\ldots, f_m$ can depend on $x_1,\ldots,x_m$ only, and all assertions follow.
\end{proof}
\begin{remark} {\em (a) This result can be refined. According e.g. to \cite{WMul}, Prop. 3.4 every set 
\[
\left\{x\in U;\,{\rm rank}\left(D\Psi(x)\right)=q\right\}
\]
is invariant for $\dot x = f(x)$ and locally a submanifold of $U$. Restriction to such sets thus suggests, and in principle allows, a case-by-case approach to reduction. \\
(b) The problem of reducibility can be "trivialized" in yet another way: By the straightening theorem, $f$ may be transformed to a constant vector field near any non-stationary point, for which reducibility (to any dimension $\geq 1$) is obvious. This observation shows that the set $U$ in Definition \ref{redudef} may play an important role. Furthermore, explicit determination of a reducing map for a given equation (or explicit determination of all vector fields reducible by a given map) is a different matter, and reducibility is a nontrivial property near stationary points (see also \cite{HaWa}).\quad$\diamond$
}
\end{remark}
The above trivialization results are of little practical relevance, due to their reliance on non-constructive theorems. Their principal value lies in providing insight into the local structure of reducing maps and reducible vector fields.

We next rephrase and generalize some results of  \cite{HaWa}, Section 2, about the correspondence between reducing maps and involution systems. It seems appropriate to start with a relatively abstract statement, to clarify the relevant properties of the underlying algebraic structures. 
To motivate the role of function algebras in the following Theorem, note that for a given reducing map $\Psi=\left(\psi_1,\ldots,\psi_r\right)^{\rm tr}$ according to Definition \ref{redudef}, the algebra of all functions $\rho(\psi_1,\ldots,\psi_r)$, $\rho$ analytic in $r$ variables, will be mapped to itself by $X_f$. Thus function algebras are naturally associated to reducing maps.

\begin{theorem}\label{basicred}  Given the analytic differential equation $\dot x=f(x)$ on $U$, let $\widetilde U\subseteq U$ be open and connected, and $A(\widetilde U)$ the algebra of analytic functions from $\widetilde U$ to $\mathbb K$.\\
{\em (a)} Let $\mathcal M$ be a Lie algebra of vector fields on $\widetilde U$ and denote by $A(\widetilde U){\mathcal M}$ the $A(\widetilde U)$-module generated by $\mathcal M$. Let $I({\mathcal M})\subseteq A(\widetilde U)$ denote the algebra of invariants of $\mathcal M$, thus $X_g(\phi)=0$ for all $g\in\mathcal M$ and all $\phi\in I({\mathcal M})$. Then for any vector field $f$ on $\widetilde U$ one has
\[
\left[f,\,{\mathcal M}\right]\subseteq A(\widetilde U){\mathcal M}\Longrightarrow X_f(I({\mathcal M}))\subseteq I({\mathcal M}).
\]
In particular, if $f$ normalizes $\mathcal M$ then $X_f$ stabilizes $I({\mathcal M})$.\\
 If, moreover, there are finitely many $\sigma_1,\ldots,\sigma_r\in I({\mathcal M})$ such that every element of $I({\mathcal M})$ can be expressed as an analytic function of the $\sigma_j$ then one obtains a reducing map $\left(\sigma_1,\ldots,\sigma_r\right)^{\rm tr}$ for $\dot x =f(x)$.\\
{\em (b)} Conversely, let $B$ be a subalgebra of $A(\widetilde U)$, and $f$ a vector field such that $X_f(B)\subseteq B$. Then $f$ normalizes the Lie algebra
\[
{\mathcal L}(B)=\left\{g;\,X_g(B)=0\right\}.
\]
Moreover ${\mathcal L}(B)$ is a module over $A(\widetilde U)$. \\
If, moreover, this module is finitely generated, say by  $g_1, \ldots, g_s$, then the $g_i$ are in involution on $\widetilde U$, thus there are $\mu_{ijk}\in A(\widetilde U)$ such that for all $i$ and $j$ relations
\[
\left[g_i,g_j\right] =\sum_k \mu_{ijk}g_k
\]
hold.
\end{theorem}
\begin{proof} To verify the nontrivial assertion of (a), let $g\in \mathcal M$ and $\psi\in I({\mathcal M})$. Then by hypothesis,
\[
0=X_{[g,\,f]}(\psi)=X_gX_f(\psi)-X_fX_g(\psi)=X_gX_f(\psi)
\]
and therefore $X_f(\psi)\in I({\mathcal M})$.
Concerning (b), note that for every $\rho\in B$ and every $g\in{\mathcal L}(B)$ one has
\[
X_{[g,f]}(\rho)=X_gX_f(\rho)-X_fX_g(\rho)=0
\]
in view of $\rho\in B$, $X_f(\rho)\in B$.
\end{proof}
\begin{remark}{\em  (a) There are obvious modifications of the Theorem for germs of local analytic functions and vector fields, resp. for polynomial and rational functions and vector fields.\\
(b) An important class of examples is formed by the systems symmetric with respect to a Lie algebra ${\mathcal M}$, thus $\left[f,\,{\mathcal M}\right]=0$.\quad$\diamond$ }
\end{remark}

Locally, the finite generation property holds in many cases, but such results are only partly constructive:
\begin{corollary}\label{finitered} Let the analytic differential equation $\dot x=f(x)$ be given on $U$, and let $y\in U$.

\noindent{\em (a)} If there is a rank $s$ analytic  involution system $g_1, \ldots, g_s$ in a neighborhood of $y$ (thus the $g_i$ are in involution and $g_1(y),\ldots, g_s(y)$ span an $s$-dimensional subspace of $\mathbb K^n$), and there exist analytic functions $\lambda_{ij}$ such that 
\[
\left[g_i,f\right] =\sum_j \lambda_{ij}g_j
\]
then there is a local reducing map $\Psi$, whose entries are common first integrals of the $g_i$, to some equation on an open subset of $\mathbb K^{n-s}$, with rank $n-s$.

\noindent{\em (b)} If there is an analytic reducing map $\Psi$ as defined in \eqref{reducond}, and $D\Psi(y)$ has maximal rank $m$, then there is an analytic involution system of rank $n-m$, defined in some neighborhood of $y$, such that the entries of $\Psi$ are common first integrals of this involution system, and the identities from part (a) hold.
\end{corollary}
\begin{proof} Part (a) is a direct consequence of Frobenius' theorem (see e.g. Olver \cite{Olv}, Section 1.3). For part (b) consider, in a suitable neighborhood of $y$, the homogeneous system of linear equations
\[
D\Psi(x)q(x)=0
\]
(over the quotient field of the ring of analytic functions) and determine a basis $g_1,\ldots,g_{n-m}$ of the solution space. Obviously one may choose a basis consisting of analytic vector fields. Since the $\left[g_i,g_j\right]$ are also solutions of this linear system, they are linear combinations of $g_1,\ldots,g_{n-m}$. The identities involving $\left[g_i,f\right]$ follow from Theorem \ref{basicred}, in view of the fact that any common first integral of the $g_i$ is locally  a function of the $\psi_j$; see Lemma \ref{loctriv}.
\end{proof}

\begin{remark} {\em The proof of part (b) shows that - in contrast to Frobenius - a corresponding involution system can be determined explicitly from the reducing map.\quad$\diamond$}

\end{remark}
To finish this subsection, we discuss the relation between reducible systems and symmetric ones.
\begin{proposition}\label{easycon}  Let the analytic differential equation $\dot x=f(x)$ be given on $U$, and assume that there is an analytic  involution system $g_1, \ldots, g_s$ in the open subset $\widetilde U$ such that $\left[f,\,g_i\right]=0$ for $1\leq i \leq s$. Then, given arbitrary analytic functions 
$\mu_1,\ldots,\mu_s$ on $\widetilde U$, the vector field
\[
f^*:=f+\sum \mu_i g_i
\]
is reducible by the common invariants of $g_1,\ldots,g_s$.
\end{proposition}
\begin{proof} For any $k$ one has 
\[
\left[g_k,\,f^*\right]= \sum_i \left(X_{g_k}(\mu_i)g_i+\mu_i\left[g_k,g_i\right]\right)
\]
due to $\left[g_k,\,f\right]=0$. The assertion follows from Theorem \ref{basicred}.
\end{proof}
Stating a converse to Proposition \ref{easycon} is not a straightforward matter. While an involution system uniquely determines the analytic invariants, the invariants do not determine a unique (finite) involution system. (Incidentally, this observation may be used to prove Frobenius' theorem; see for instance Hermann \cite{Her}.) But the following statement holds.
\begin{proposition}\label{easyconcon}
Let $g_1,\ldots, g_s$ form an analytic involution system on an open set $U$ such that the vector field $f^*$ is reducible by the common invariants of the $g_i$. Then locally, near any point of $U$ where $(g_1,\ldots, g_s)$ has maximal rank $s$, the module spanned by the $g_i$ has a basis $\hat g_i=\sum_j \sigma_{ij}g_j$, with analytic functions $\sigma_{ij}$, $1\leq i,\,j\leq s$, and there exist analytic functions $\mu_i$ such that $f:=f^*-\sum\mu_j\hat g_j$ satisfies $\left[f,\,\hat g_i\right]=0$ for all $i$.
\end{proposition}
\begin{proof}
Use Lemma \ref{loctriv}, with invariants $x_1,\ldots,x_m$ and module basis $e_{m+1},\ldots, e_n$. Then (with the same notation as in the proof of Lemma \ref{loctriv})
\[
f^*(x)=\left(\begin{array}{c}f_1(x_1,\ldots,x_m)\\
                                               \vdots\\
                                             f_m(x_1,\ldots,x_m)\\
                                                   0\\
                                              \vdots\\
                                                  0\end{array}\right)+\left(\begin{array}{c}0\\
                                               \vdots\\
                                             0\\
                                                   \ast\\
                                              \vdots\\
                                                  \ast\end{array}\right)
\]
provides the asserted decomposition.
\end{proof}
\begin{corollary}\label{easyconcor} In the special case $s=1$, $g=g_1$
a vector field $f^*$ is reducible by the invariants of $g$ if and only if
there is a vector field $f$ and a scalar function $\rho$ such that $\left[g,\,f\right]=0$ and $f^*=f+\rho g$.
\end{corollary}
\begin{proof} Using the notation of Proposition \ref{easycon} we have $\hat g=\sigma g$ and $f=f^*-\mu \hat g$; the assertion follows with $\rho
=\mu\sigma$.
\end{proof}
\begin{remark}\label{easyconrem}{\em (a) In the special case $s=1$, $g=g_1$, a direct proof  of Proposition \ref{easyconcon} runs as follows: Assume $\left[g,\,f^*\right]=\beta g$ for some function $\beta$, and make the ansatz $f=f^*-\mu \,g$. Then $\left[g,\,f\right]=0$ if and only if
$X_g(\mu)+\beta =0$, and the latter has a solution near any non-stationary point of $g$ (e.g. by the straightening theorem). Thus the function $\mu$ can be determined explicitly whenever a straightening map for $g$ is explicitly known.
\\
(b) As noted above, passing to a different module basis $\hat g_i=\sum \sigma_{ij}g_j$ will not change the reducibility conditions and properties but may affect other distinguished properties, like commutation of vector fields. For the case of one module generator $g=g_1$ this was discussed in Pucci and Saccomandi \cite{PuSa}.\quad$\diamond$
}
\end{remark}
\subsection{Orbital reducibility}
For first-order ordinary differential equations one is not only interested in symmetries but more generally in orbital symmetries. Similarly, it is sensible to generalize from reduction to orbital reduction, as we will do next. We first recall a characterization of orbital symmetry from \cite{WMul}, Lemma 2.3.
\begin{proposition}\label{altorb}
Let the analytic differential equation $\dot x=f(x)$ be given on $U$. Assume that there is a rank $m$ analytic  involution system $g_1, \ldots, g_s$ in the open subset $\widetilde U$, with $n-m$ independent common invariants $\psi_1,\ldots,\psi_{n-m}$ such that every common invariant can be expressed as an analytic function of the $\psi_j$. Assume that $f$ is not an element of the module generated by the $g_j$. Then \eqref{ode} is orbitally symmetric with respect to $g_1,\ldots,g_s$ if and only if there is an analytic function $\mu$ without zeros on $\widetilde U$ such that 
\[
\left[\mu f,\,g_j\right]=0,\quad 1\leq j\leq s.
\]
\end{proposition}
Correspondingly, we define: 
\begin{Definition}\label{orbredudef}
Equation \eqref{ode} is {\em orbitally reducible} by the map $\Psi$ if some equation $\dot x = \mu(x)\cdot f(x)$ ($\mu$ analytic without zeros on an open subset $\widetilde U\subseteq U$) is reducible by $\Psi$; in other words, instead of \eqref{reducond} the identity 
\begin{equation}\label{orbreducond}
\mu(x)D\Psi(x)f(x)=h(\Psi(x))
\end{equation}
holds on $\widetilde U$. 
\end{Definition}
\begin{remark}{\em (a) This is clearly a necessary and sufficient condition for solution orbits of \eqref{ode} to be mapped by $\Psi$ to solution orbits of $\dot x=h(x)$. Equivalently, for every local first integral $\rho$ of $\dot x=h(x)$ the pullback $\rho\circ \Psi$ is a first integral of $\dot x=f(x)$.\\
(b) Via autonomization this definition extends to non-autonomous equations, and due to part (a) and Remark \ref{autonon} this is the natural notion of a reducing map for non-autonomous equations; i.e., some map which sends solutions of a system to solutions of a system "in smaller dimension".\quad$\diamond$}
\end{remark}
\begin{remark}\label{trivorbred}{\em While explicitly finding a reducing map to a one-dimensional equation is a hard problem, finding an orbital reducing map to any one-dimensional equation with nonzero right-hand side is locally trivial: For any $h$ such that $h\circ\Psi$ is not identically zero, one may choose the factor $\mu$ in a suitable way. But this observation is of little interest since it provides no information about solutions to \eqref{ode}.\quad$\diamond$}
\end{remark}
Next we want to give a characterization of orbital reducibility in terms of Lie bracket properties. The critical argument in one direction of the proof is similar to \cite{WMul}, Lemma 2.3.
\begin{theorem}\label{orborder}
Let the analytic differential equation $\dot x=f(x)$ be given on $U$, and let a rank $m$ analytic  involution system $g_1, \ldots, g_s$ be given in the open subset $\widetilde U$. Assume that there are $n-m$ independent common invariants $\psi_1,\ldots,\psi_{n-m}$ of the $g_i$ such that every common invariant can be expressed as an analytic function of the $\psi_j$. Then \eqref{ode} is orbitally reducible by the $\psi_j$ if and only if there are analytic functions $\alpha_i$ and $\lambda_{ij}$ on $\widetilde U$ such that
\begin{equation}\label{orbred}
\left[g_i,f\right] =\alpha_i f+\sum_j \lambda_{ij}g_j,\quad 1\leq i\leq s
\end{equation}
\end{theorem}
\begin{proof}
 In the setting of Corollary \ref{finitered}  we obtain 
\[
\left[\mu \cdot f,\,g_i\right]=\sum_j \gamma_{ij}g_j,\quad 1\leq i\leq s, 
\]
and therefore
\[
\mu \cdot\left[ f,\,g_i\right]=X_{g_i}(\mu)\cdot f+\sum_j \gamma_{ij}g_j,\quad 1\leq i\leq s.
\]
Division by $\mu$ yields \eqref{orbred}.

For the reverse direction assume that 
\eqref{orbred}
holds on $U$.
If $X_f(\phi)=0$ for every common invariant $\phi$ of the $g_j$ then $f$ lies in the module generated by the $g_j$, due to the rank and independence conditions, and the bracket relation holds trivially. 
Otherwise, let $\psi$ be analytic such that
\[
X_{g_1}(\psi)=\cdots=X_{g_s}(\psi)=0,\mbox{   but   }X_f(\psi)\not=0.
\]
Setting 
\[
f^*:=\frac 1{X_f(\psi)}f
\]
one obtains
\[
\left[g_i,f^*\right] =\sum_j \lambda^*_{ij}g_j,\quad 1\leq i\leq s
\]
with analytic $ \lambda^*_{ij}$ on $U^*:=\{x\in U:\,X_f(\psi)(x)\not=0\}$.
Indeed, for any $i$ the commutation relation for $g_i$ and $f$ implies
\[
X_{g_i}X_f(\psi)-X_fX_{g_i}(\psi)=\alpha_iX_f(\psi)+\sum_j\lambda_{ij}X_{g_j}(\psi),
\]
and thus
\[
X_{g_i}X_f(\psi)=\alpha_iX_f(\psi).
\]
In view of
\[
\left[g_i,\frac 1{X_f(\psi)}f\right]=\frac 1{X_f(\psi)}\left[g_i,f\right]-\frac {X_{g_i}X_f(\psi)}{X_f(\psi)^2}f
\]
the assertion follows with $\lambda_{ij}^*:=\lambda_{ij}/{X_f(\psi)}$.
\end{proof}

\begin{remark} \label{onegee}
 {\em If the involution system consists just of $g=g_1$ then one obtains, as a particular case, the condition $\left[g,\,f\right]=\alpha f + \lambda g$. This leads to the $\lambda$-symmetries of Muriel and Romero; see \cite{MuRo1,MuRo1a,MuRo2, MuRo3}.  Theorem 2.5 in \cite{MuRo1a} corresponds directly to the above Theorem in case $s=1$; see also Section 2 of \cite{MuRo3}. The theoretical framework was clarified by Morando \cite{Mor} (Subsection 4.2 in particular). These observations presume the identification of higher-order ODEs with an equivalent first-order system, and restriction of $g$ to infinitesimal point transformations. (See more on this in Section \ref{Higher} below.) For this reason one could call $(g_1,\ldots,\,g_s)$ a system of {\em joint $\lambda$-symmetries} for the equation \eqref{ode}. In \cite{CGWsprol} the name {\em $\sigma$- symmetries} was chosen, to emphasize the focus on prolongations. \quad$\diamond$}
\end{remark}
\begin{remark}\label{orbredred}
 {\em In the setting of Theorem \ref{orborder} the actual computation of an orbitally reduced system works as follows. There exists some $\mu$ such that the identities
\[
\left[\mu \cdot f,\,g_i\right]=\sum_j \gamma_{ij}g_j,\quad 1\leq i\leq s, 
\]
hold, and by Corollary \ref{finitered} there exist analytic functions $h_i$ such that
\[
\mu X_f(\psi_i)=X_{\mu f}(\psi_i)=h_i(\psi_1,\ldots,\psi_{n-m})
\]
for $1\leq i\leq n-m$. While $\mu$ may not be explicitly known, one may turn to $h_i^*:=h_i/h_1$ for $1\leq i\leq n-m$, which can be expressed as functions of the $\psi_i$ alone. Therefore the $\psi_i$ define an orbit-preserving map from \eqref{ode} to $\dot y=h^*(y)$.\quad$\diamond$
}
\end{remark}
We finish this subsection with the counterparts to Proposition \ref{easycon} ff. The proofs are immediate, in view of Theorem \ref{orborder}, Proposition \ref{altorb}, Proposition \ref{easycon}, Proposition \ref{easyconcon} and Corollary \ref{easyconcor}.
\begin{proposition}\label{easyconorb}  Let the analytic differential equation $\dot x=f(x)$ be given on $U$, and assume that there is an analytic  involution system $g_1, \ldots, g_s$ and analytic functions $\alpha_i$ in the open subset $\widetilde U$ such that $\left[f,\,g_i\right]=\alpha_i f$ for $1\leq i \leq s$. Then, given arbitrary analytic functions 
$\mu_1,\ldots,\mu_s$ on $\widetilde U$, the vector field
\[
f^*:=f+\sum \mu_i g_i
\]
is orbitally reducible by the common invariants of $g_1,\ldots,g_s$.
\end{proposition}

\begin{proposition}\label{easyconorbcon}
Let $g_1,\ldots, g_s$ form an analytic involution system on an open set $U$ such that the vector field $f^*$ is orbitally reducible by the common invariants of the $g_i$. Then locally near any point of $U$ where $(g_1,\ldots, g_s)$ has maximal rank $s$, the module spanned by the $g_i$ has a basis $\hat g_i=\sum_j \sigma_{ij}g_j$, with analytic functions $\sigma_{ij}$, $1\leq i\leq s$, and there exist analytic functions $\mu_i$ such that $f:=f^*-\sum\mu_j\hat g_j$ satisfies $\left[f,\,\hat g_i\right]=\alpha_i f$ with suitable analytic $\alpha_i$, for $1\leq i\leq s$.
\end{proposition}
\begin{corollary}\label{easyconorbcor} In the special case $s=1$, $g=g_1$
a vector field $f^*$ is orbitally reducible by the invariants of $g$ if and only if
there is a vector field $f$ and scalar functions $\alpha$, $\rho$ such that $\left[g,\,f\right]=\alpha f$ and $f^*=f+\rho g$.
\end{corollary}

\subsection{Reduction by group invariants}
In this subsection we consider some aspects of the "inverse problem" to determine all vector fields that are (orbitally) reducible by some given map.
Following a guiding principle established in Olver and Rosenau \cite{OlRo2} (albeit in the context of partial differential equations), we consider reduction of \eqref{ode} by invariants of some group, with the system itself not necessarily symmetric.
Given a (local Lie) group of transformations, it may be difficult to determine all differential equations which are symmetric with respect to this group. By extension it it may be difficult to determine all differential equations which are reducible by its invariants.  But at least Propositions \ref{easycon} and \ref{easyconorb} provide a simple construction of reducible systems from symmetric ones, and for local one-parameter groups this construction yields all reducible systems. We extend this result and show that for Lie algebras of compact and connected (linear) Lie groups, there is a method to construct all reducible vector fields on an open and dense subset. The underlying reason is the existence of a convenient representation (on an open-dense subset) for any vector field. This may be considered a consequence of the slice theorem (see e.g. Br\"ocker and tom Dieck \cite{BtD}), but we use a simple shortcut. The following results are an extension of \cite{HaWa}, Example 2.5, where the Lemma and the first part of the Proposition were proven. The remaining statements are clear from the previous subsections.
\begin{lemma}\label{compredprep} Let $G\subseteq GL(n,\,\mathbb R)$ be a connected compact linear group, with invariant scalar product $\left<\cdot,\,\cdot\right>$, and denote by $\mathcal G$ its Lie algebra. \\
Denote by $s$ the maximal orbit dimension of $G$ and let $B_1,\ldots,B_s\in{\mathcal G}$ be such that $B_1z,\ldots,B_sz$ are linearly independent in $\mathbb R^n$ for some $z$, hence for all $z$ in an open-dense subset. (In other words, $s$ is the rank of the involution system generated by ${\mathcal G}$.) Then there exist algebraically independent polynomial invariants $\sigma_1,\ldots,\sigma_{n-s}$; let their gradients $q_j$ be defined by
\[
D\sigma_j(x)y=\left<q_j(x),\,y\right>.
\]
The $q_j$ are $G$-symmetric (thus every transformation in $G$ is a symmetry for $\dot x = q_j(x)$), moreover
\[
\theta(x):=\det\left(B_1x,\ldots,\,B_sx,\,q_1(x),\ldots,\,q_{n-s}(x)\right)
\]
is a nonzero polynomial, and every vector field $f$ on $U$ admits a representation
\begin{equation}\label{specrep}
f(x)=\sum\alpha_i(x)\,B_ix+\sum \beta_j(x)q_j(x)
\end{equation}
which holds on $\widetilde U:=\left\{x\in U;\,\theta(x)\not=0\right\}$.

\end{lemma}
\begin{proposition}\label{compred} Let the hypotheses and notation be as in Lemma \ref{compredprep}, and let the vector field $f$ be represented as in \eqref{specrep}.\\
{\em (a)} The vector field $f$ is reducible by the invariants of $G$ if and only if all $\beta_j$ are $G$-invariant. This is equivalent to 
\[
\hat f(x):=\sum \beta_j(x)q_j(x)
\]
being $G$-symmetric. Thus every reducible system on a subset of $\widetilde U$ is obtained from a symmetric one via Proposition \ref{easycon}.\\
{\em (b)} The vector field $f$ is orbitally reducible by the invariants of $G$ if and only if there is an analytic $\nu$ such that all $\beta_j=\nu\cdot\widetilde \beta_j$, with $\widetilde \beta_j$ $G$-invariant. This is equivalent to 
\[
\hat f(x):=\sum \beta_j(x)q_j(x)
\]
being orbitally $\cal G$-symmetric. Thus every orbitally reducible system on a subset of $\widetilde U$ is obtained from an orbitally symmetric one via Proposition \ref{easyconorb}.
\end{proposition}
\begin{remark} {\em (a) Note that Lemma \ref{compredprep} and Proposition \ref{compred} provide a construction of reducible systems which does not require a priori knowledge of all symmetric systems. (Actually, finding all reducible systems here is less troublesome than finding all symmetric systems.) \\
(b) The polynomial $\theta$ is not necessarily $G$-invariant (contrary to the statement in \cite{HaWa}) but generally $\theta$ is the product of a $G$-invariant polynomial $\theta_1$ and  a polynomial $\theta_2$ with the property that $B_1v,\ldots,B_sv$ are linearly dependent for all zeros $v$ of $\theta_2$ (in $\mathbb C^n$).\quad$\diamond$}
\end{remark}
\noindent{\em Examples.} (a) Consider in $\mathbb R^3$ the system
\[
\dot x=f(x)=\alpha(x)\left(
\begin{array}{l}
x_2\\
-x_1\\
0
\end{array}\right)+\beta_1(x)\left(
\begin{array}{l}
x_1\\
x_2\\
0
\end{array}\right)+\beta_2(x)\left(
\begin{array}{l}
0\\
0\\
x_3
\end{array}\right).
\]
This representation corresponds to Proposition \ref{compred}(a), with the 1-dimensional subgroup $G$ of
$SO(3,R)$ generated by $g=(x_2,-x_1,0)$, on the set $\widetilde U$ defined by $(x_1^2+x_2^2)\cdot x_3\not=0$.
The 
system is $G$-symmetric if, and only if, $\alpha$ and the $\beta_i$ are functions of the generating invariants $\sigma _1=x_1^2+x_2^2$ and $\sigma_2=x_3$ alone. (This holds because the Lie algebra is abelian. Generally finding symmetric systems is a more involved matter.) The system is
reducible by the invariants $\sigma_1,\,\sigma_2$ of $g$ if  $\beta_1$ and $\beta_2$
are  functions of $x_1^2+x_2^2$ and $x_3$ alone , with $\alpha$  arbitrary.
Assuming e.g. $\beta_2\not=0$, the system is orbitally symmetric with respect to $g$ if and only if $\alpha/\beta_2$ and $\beta_1/\beta_2$ are functions of $\sigma_1$ and $\sigma_2$ alone, and the system is orbitally
reducible by the invariants of $g$ if and only if $\beta_1/\beta_2$ is a function of $\sigma_1$ and $\sigma_2$ alone.

\medskip
\noindent (b) We consider $G=SO(3,\mathbb R)$. On $\mathbb R^3$ the linear maps
\[
B_1(x)=\left(\begin{array}{c}-x_2\\
                                               x_1\\
                                                0\end{array}\right),\quad
B_2(x)=\left(\begin{array}{c}-x_3\\
                                               0\\
                                                x_1\end{array}\right),\quad
B_3(x)=\left(\begin{array}{c}0\\
                                              - x_3\\
                                                x_2\end{array}\right)
\]
span the Lie algebra $\cal G$, one has $s=2$, and may choose $B_1$ and $B_2$, since $B_1z$ and $B_2z$ are linearly independent for all $z\in\widetilde U:=\left\{x;\,x_1\not=0\right\}$.
The invariant algebra is generated by the polynomial $\sigma(x):=x_1^2+x_2^2+x_3^2$, and its gradient is equal to $q(x)=2x$. 
One has $\theta(x)=2x_1(x_1^2+x_2^2+x_3^2)$ and on the set given by $\theta\not=0$ any vector field can be written in the form
\[
f(x)=\alpha_1(x)\cdot B_1x+\alpha_2(x)\cdot B_2x+\beta(x)\cdot x.
\]
This vector field is reducible by $\sigma$ if and only if $\beta$ is group-invariant, thus can be written as a function of $\sigma$ alone. There is no restriction on orbital reducibility by $\sigma$, as was to be expected from Remark \ref{trivorbred}.

\section{Higher order equations}\label{Higher}
In this section we apply the results on first-order systems to ordinary differential equations of higher order, and thus gain a new perspective on the construction and reduction of equations admitting lambda symmetries. Recall the correspondence between higher-order equations and systems of first order:
Given a (single) non-autonomous equation of order $m+1>1$,
\begin{equation}\label{high}
x^{(m)}= p(t,x,\dot x,\ldots,x^{(m-1)}), 
\end{equation}
its solutions correspond to solutions of the first-order system
\begin{equation}\label{degonesys}
\begin{array}{lll}
\dot x_1&=& x_2\\
             &\vdots & \\
\dot x_{m-1}&=& x_m\\
\dot x_m&=& p(t,x_1,x_2,\ldots,x_m).
\end{array}
\end{equation}
Therefore symmetries of the first-order system will send solutions of \eqref{high} to solutions of \eqref{high}. In other words, orbital symmetries of the autonomous system
\begin{equation}\label{autodegonesys}
\begin{array}{lll}
\dot x_0&=&1\\
\dot x_1&=& x_2\\
             &\vdots & \\
\dot x_{m-1}&=& x_m\\
\dot x_m&=& p(x_0,x_1,x_2,\ldots,x_m),
\end{array}\quad \mbox{briefly   }\dot x = P(x),
\end{equation}
will send solutions of \eqref{high} to solutions of \eqref{high} (up to familiar identifications). This point of view is proposed in the monograph by Stephani \cite{Ste}. It may be worth noting (and has already been noted in special instances, e.g. by Nucci and Leach \cite{NuLea}) that conversely any first order system may locally be represented as a single higher order equation.
\begin{proposition}\label{raiseord}
Let a nonautonomous first-order system
\[
\dot z = q(t,z) \mbox{   on  }\widetilde U\subseteq \mathbb K\times \mathbb K^m
\]
be given, and let $(t^*,\,z^*)$ such that $q(t^*,\,z^*)\not=0$. Then there exist local coordinates $t,\,x_1,\ldots,\,x_m$ near $(t^*,\,z^*)$ in which the system takes the form \eqref{degonesys}.
\end{proposition}
\begin{proof} With no loss of generality we have $t^*=0$ and $z^*=0$. Consider the autonomized system
\[
\frac{d}{dt}\left(\begin{array}{c}t\\
                                        z\end{array}\right)=\left(\begin{array}{c}1\\
                                        q(t,z)\end{array}\right)=:Q(t,\,z).
\]
We may assume that $Q(0,\,0)=(1,\,1,\,0,\ldots,\,0)^{\rm tr}$. By a straightforward variant of the straightening theorem we may furthermore assume that locally $Q(t,\,z)=(1,\,1,\,0,\ldots,\,0)^{\rm tr}$ is constant.
Now let 
\[
\phi(t,z):=\sum_{j=2}^m\frac{1}{(j-2)!}z_1^{j-2}z_j + \frac{1}{m!}z_1^m
\]
and define 
\[
x_k:=\frac{d^{k-1}\phi}{dz^{k-1}}=X_Q^{k-1}(\phi),\quad 1\leq k\leq m.
\]
At $z=0$ the Jacobian matrix of $(x_1,\ldots,\,x_m)^{\rm tr}$ as a function of $z$ equals
\[
\left(
\begin{array}{cccccc}0&1&0&\cdots&0&0\\
                                0&0&1&\cdots&0&0\\
                                    & & \vdots& & &\\
                                 0&0&0&\cdots& 0&1\\
                                  1&0&0&\cdots&0&0 \end{array}\right)
\]
and therefore we have a coordinate transformation $(t,\,z)\mapsto (t,x)$. By design one has $\dot x_j=x_{j+1}$ for $1\leq j\leq m-1$, and this proves the assertion.
\end{proof}
\begin{remark}\label{raiseordrem}{\em (a) The proof is constructive, to some point, even if $Q$ is given in general form (only assuming $Q(0,\,0)\not=0$). Take any function $\phi$ such that $\phi,\,X_Q(\phi),\,\ldots,X_Q^{m-1}(\phi)$ are
functionally independent, and choose the new variables accordingly. The above proof amounted to verifying the existence of such a function, and also makes clear that "almost every" function will satisfy this property. More precisely, in the local ring of analytic functions those $\phi$ which do not satisfy the independence property form a subset of positive (finite) codimension.\\
(b) The exceptional case of functionally dependent $\phi,\,X_Q(\phi),\,\ldots,X_Q^{m-1}(\phi)$ provides reduction in a direct manner. Let $\ell$ be maximal such that $x_1:=\phi,\,x_2:=X_Q(\phi),\,\ldots,x_\ell:=X_Q^{\ell-1}(\phi)$ are independent. Then $X_Q^{\ell}(\phi)$ may be expressed as a function of $x_1,\ldots, x_\ell$, whence one has reduction to an equation of order $\ell+1$ (or the associated system).\\
(c) This Proposition opens, in principle, a possibility to determine symmetries of first-order systems: Rewrite the system as a higher-order equation, and determine the point symmetries of the latter, which amounts to a fully algorithmic procedure. (Nucci and Leach \cite{NuLea} noted and used a variant of such an approach.) The drawback is, of course, that one will only find trivial symmetries in general. It could be interesting to explore the possibility of a systematic approach, which would include the question how to choose a suitable function $\phi$.\quad$\diamond$
}
\end{remark}
\subsection{Review: Prolongation in simple cases}

In most monographs on symmetries, such as Olver \cite{Olv}, or Krasil'shchik and Vinogradov \cite{KraVi }, the Lie point symmetries of higher order equations \eqref{high} are determined via the general procedure for prolongations of vector fields to jet spaces. The following shortcut works and is quite useful in our scenario; it is essentially taken from Stephani \cite{Ste}, Ch.~I, Section 3.5.

We are interested in vector fields $g$ on $\widetilde U$ such that $\left[g,q\right]=\mu\cdot q$ for some $\mu$. By geometric motivation one frequently considers only {\em infinitesimal point symmetries} of the higher order equation, which implies the projectability property
\begin{equation}\label{pointtrafo}
g(x)=\left(\begin{array}{l}
g_0(x_0,x_1)\\
g_1(x_0,x_1)\\
g_2(x_0,x_1,x_2)\\
\,\,\,\vdots\\
g_{m-1}(x_0,x_1,\ldots,x_{m-1})\\
g_m(x_0,x_1,\ldots,x_m)
\end{array}\right)
\end{equation}
for the infinitesimal symmetry of the associated system. Step-by-step evaluation of the condition $\left[g,q\right]=\mu\cdot q$ now yields
\[
\begin{array}{rclcl}
 & -& \partial_0g_0-x_2\cdot \partial_1g_0 &=&\mu\\
g_2&-&\partial_0g_1-x_2\cdot \partial_1g_1 &=&\mu x_2\\
g_3&-&\partial_0g_2-x_2\cdot \partial_1g_2-x_3\partial_2g_3 &=&\mu x_3\\
 & &  & \vdots& \\
g_m&-&\partial_0g_{m-1}-x_2\cdot \partial_1g_{m-1}-\ldots -x_m\partial_mg_{m-1} &=&\mu x_m\\
\end{array}
\]
Thus one may successively compute $\mu$, $g_2,\ldots, g_m$ from $g_0$ and $g_1$ and their derivatives. Evaluating the last entry of the Lie bracket, which has not been written down here, provides an overdetermined,  algorithmically accessible, system of partial differential equations for $g_0$ and $g_1$, and thus in effect the symmetry conditions. (This approach is of course equivalent to the usual prolongation procedure for point symmetries of \eqref{high}.)

\medskip
Generalizing to orbital reducibility, but keeping the geometric restriction \eqref{pointtrafo}, one deals with the lambda symmetries first considered by Muriel and Romero \cite{MuRo1}. The condition is 
\[
\left[g,q\right]=\mu\cdot q + \lambda\cdot g.
\]
Initially Muriel and Romero  \cite{MuRo1} (for geometric reasons) require $\lambda$ to be a function (possibly a priori unknown) of $x_0,\,x_1,\, x_2$ only. One obtains the "$\lambda$-prolongation formulas":
\[
\begin{array}{rclcl}
 & -& \partial_0g_0-x_2\cdot \partial_1g_0 &=&\mu + \lambda g_0\\
g_2&-&\partial_0g_1-x_2\cdot \partial_1g_1 &=&\mu x_2+ \lambda g_1\\
g_3&-&\partial_0g_2-x_2\cdot \partial_1g_2-x_3\partial_2g_3 &=&\mu x_3+ \lambda g_2\\
 & &  & \vdots& \\
g_m&-&\partial_0g_{m-1}-x_2\cdot \partial_1g_{m-1}-\ldots -x_m\partial_mg_{m-1} &=&\mu x_m+ \lambda g_{m-1}\\
\end{array}
\]
In this setting,  $\lambda$, $g_0$ and $g_1$ successively determine $\mu$, $g_2,\ldots, g_m$, and again the last entry of the Lie bracket identity will provide compatibility conditions, as was noted by Muriel and Romero \cite{MuRo1}. The determination of $\lambda$-symmetries is not a completely algorithmic procedure, and the artwork in this approach is to suitably determine (e.g. by educated guesses) $\lambda$ such that prolongation and evaluation yields nontrivial results.  Muriel and Romero's restriction imposed on $\lambda$ is of importance  for the construction of higher order differential invariants further on in \cite{MuRo1}. One should also note that these authors relaxed the geometric restrictions on $\lambda$ in a subsequent paper \cite{MuRo2}, thus making another step towards general orbital reduction with respect to a single vector field, albeit in the setting of a higher order equation. (See also Remark \ref{onegee}.)

In the following discussion and construction of reducible higher-order differential equations we emphasize the correspondence to first order systems. 

\subsection{Constructing reducible higher-order equations}
We propose a construction of reducible higher-order equations, thus extending the work by Muriel and Romero \cite{MuRo1,MuRo1a, MuRo2, MuRo3} on lambda symmetries. 
We use Theorem \ref{orborder} and Proposition \ref{easyconorb}, based on the correspondence to first-order systems established in Proposition \ref{raiseord}, starting from a system with known orbital symmetries. Thus, on the one hand,  as in Muriel and Romero \cite{MuRo2} the "lambdas" may depend on all variables, and moreover we do not necessarily restrict attention to point symmetries of higher order equations. On the other hand, we extend the framework of Muriel and Romero by considering more than one infinitesimal orbital symmetry. As noted earlier, we therefore work in the most general setting for (orbital) reducibility of the associated first order system. In contrast to Muriel and Romero \cite{MuRo1,MuRo2,MuRo3} our focus is on constructing reducible equations, rather than detecting reducibility in given equations. We emphasize that the following should be seen only as a first step towards a systematic construction of reducible equations with prescribed (joint) $\lambda$-symmetries.

Consider a single equation \eqref{high} of order $m+1$ and rewrite it as the autonomized system \eqref{autodegonesys}:
\[
\begin{array}{lll}
\dot x_0&=&1\\
\dot x_1&=& x_2\\
  &\vdots&\\
\dot x_{m-1}&=& x_m\\
\dot x_m&=& p(x_0,x_1,\ldots,x_m)
\end{array}
\]
Assume that for $k=1,\ldots,r$ the system admits infinitesimal symmetries
\[
g^{(k)}=\left(\begin{array}{c} g_0^{(k)}(x_0,x_1,\ldots,x_m)\\
                                       g_1^{(k)}(x_0,x_1,\ldots,x_m)\\
                                          \vdots \\
                                      g_m^{(k)}(x_0,x_1,\ldots,x_m)
\end{array}\right)
\]
which form an involution system. Then for any system of scalar functions $\nu  ^{(k)}$ the equation
\[
\dot x=\widehat H(x):=\left(\begin{array}{c} 1+\sum_k \nu^{(k)} g_0^{(k)}(x)\\
                                       x_2+\sum_k \nu^{(k)} g_1^{(k)}(x)\\
                                      \vdots     \\
                                         x_m+\sum_k \nu^{(k)} g_{m-1}^{(k)}(x)\\
                                      p(x)+\sum_k \nu^{(k)} g_m^{(k)}(x)
\end{array}\right)
\]
is orbitally reducible by the common invariants of the $g^{(k)}$.
Now consider the orbitally equivalent system 
\[
\dot x=H(x):=\frac1{1+\sum\nu^{(k)} g_0^{(k)}(x)}\widehat H(x)=:\left(\begin{array}{c} 1\\
                                       h_1(x)\\
                                      \vdots \\
                                        h_m(x)
\end{array}\right).
\]
and introduce new coordinates 
\[
t=x_0,\,y_1=x_1,\, y_2=h_1(x),\ldots
\]
which is always possible by Proposition \ref{raiseord}. In general one will thus obtain an equation of order $m+1$ for $y=y_1$. (The exceptional case when the functions $x_1,\,X_H(x_1),\ldots, X_H^{m-1}(x_1)$ are not independent implies reducibility via Remark \ref{raiseordrem}.) By construction, the equation $\dot x=H(x)$
is orbitally reducible by the common invariants of the $g^{(k)}$, and in the generic case this holds true for the system in new coordinates
$t,y_1,\ldots, y_m$ (with the invariants also written in new coordinates), and for the corresponding equation of order $m+1$. Thus we have constructed reducible higher order equations from symmetric ones. Note that if the coordinate transformation is given by $x=\Psi(y)$ then the "joint-$\lambda$ symmetries" in new coordinates are given by
\[
\tilde g^{(k)}(y)=D\Psi(y)^{-1}g^{(k)}(\Psi(y)).
\]

\begin{remark}{\em The special case when $r=1$ and the geometric restrictions on $g=g_1$ (and $\lambda$) hold is, naturally, of particular interest. It is possible to explicitly construct all $\lambda$-symmetric higher-order equations which are reducible by the invariants of $g$, assuming the latter are known. First, one knows that all symmetric higher-order equations are given by functions of the differential invariants of $g$ (with appropriate identifications; see e.g. Olver \cite{Olv}, Ch.~2), and second, by Corollary \ref{easyconorbcor} one can construct all reducible equations via the procedure outlined above. Here one should assume $\nu$ to be a function of $x_0$, $x_1$ and $x_2$ only, so that the same holds for  $\lambda$. \quad $\diamond$
}
\end{remark}

\begin{remark}\label{lamfind}{\em In the special case $r=1$ it is also of interest to identify the "$\lambda$" emerging from this procedure. Thus start with vector fields $f$ and $g$ such that $\left[g,\,f\right]=\alpha f$ with some scalar function $\alpha$. Given a scalar function $\nu$, form
\[
\widehat H(x) = f(x)+\nu g(x);\quad H(x)=\frac1{1+\nu g_0(x)}\widehat H(x).
\]
Straightforward computations show that
\[
\left[ g,\widehat H\right]=\alpha\widehat H+\left(X_g(\nu)-\alpha \nu\right)g; \quad \left[ g, H\right]=\left(\dots\right) H+\frac{\left(X_g(\nu)-\alpha \nu\right)}{1+\nu g_0}g;
\]
and the coefficient of $g$ in the second identity (rewritten in new coordinates via $x=\Psi(y)$) is the scalar function $\lambda$ as introduced in \cite{MuRo1}.
 \quad $\diamond$
}
\end{remark}

\subsection{Examples of order two}
Since orbital reduction to dimension one is of little interest (recall Remark \ref{trivorbred}), "joint $\lambda$" is of little interest here. Therefore we remain mostly within the framework of Muriel and Romero \cite{MuRo1,MuRo2}, considering the inverse problem of finding differential equations with prescribed reduction.
Rewrite a single second order equation as an autonomous system:
\[
\begin{array}{lll}
\dot x_0&=&1\\
\dot x_1&=& x_2\\
\dot x_2&=& p(x_0,x_1,x_2)
\end{array}
\]
Assume that 
\[
g=\left(\begin{array}{c} g_0(x_0,x_1,x_2)\\
                                       g_1(x_0,x_1,x_2)\\
                                      g_2(x_0,x_1,x_2)
\end{array}\right)
\]
is an infinitesimal orbital symmetry for this equation. Then for any scalar function $\nu$ the equation
\[
\dot x=\widehat H(x):=\left(\begin{array}{c} 1+\nu g_0(x)\\
                                       x_2+\nu g_1(x)\\
                                      p(x)+\nu g_2(x)
\end{array}\right)
\]
is orbitally reducible by the invariants of $g$, and every orbitally reducible system is obtained from an orbitally symmetric one in this way, due to Proposition \ref{easyconorb} and Corollary \ref{easyconorbcor}.
The orbitally equivalent system 
\[
\dot x=H(x):=\frac1{1+\nu g_0(x)}\widehat H(x)=:\left(\begin{array}{c} 1\\
                                       h_1(x)\\
                                      h_2(x)
\end{array}\right)
\]
remains orbitally reducible by the invariants of $g$. Now, unless $h_1$ depends on $x_0$ and $x_1$ alone, there is a local coordinate change
\[
\begin{array}{rcl} y_0&=& x_0\\
                              y_1&=& x_1\\
                               y_2&=& X_H(y_1)=h_1(x)
\end{array}
\]
and the system in new coordinates
\[
\dot y= H^*(y)=\left(\begin{array}{c} 1\\
                                       y_2\\
                                      p^*(y)
\end{array}\right)
\]
is orbitally reducible by the invariants of $g$, expressed in new coordinates. Thus we have obtained a reducible second-order equation from a symmetric one, and every equation which is orbitally reducible by the invariants of $g$ is obtained in this way. Generally, the method is not completely constructive, since an explicit computation of $p^*$ requires an explicit inverse to the coordinate transformation. Therefore we (have to) make special choices of functions in the concrete examples below. 

\medskip
\noindent{\it Example 1.} Consider the system
\begin{equation}\label{ex1}
\begin{array}{lll}
\dot x_0&=&1\\
\dot x_1&=& x_2 \\ 
\dot x_2&=& \gamma(x_2)/x_1
\end{array}
\end{equation}
Here $\gamma$ is an arbitrary analytic function of one variable.
This system admits the infinitesimal orbital symmetry
\[
g=\left(\begin{array}{c}  x_0\\
                                        x_1\\
                                       0
\end{array}\right)
\]
(which corresponds to a Lie point symmetry of the associated second-order equation). An independent set of invariants of $g$ is given by 
\[ \sigma_1=x_1/x_0 \quad\quad  \sigma_2=x_2.\]
Therefore, given any scalar function $\nu$ the equation
\[
\dot x=\widehat H(x):=\left(\begin{array}{c} 1+\nu x_0\\
                                       x_2+\nu x_1\\
                                      \gamma(x_2)/x_1
\end{array}\right)
\]
is orbitally reducible by these invariants,
and indeed one has:
\[ x_0 \dot\sigma_1:= x_0\cdot X_{\widehat H}(\sigma_1)=(\sigma_2-\sigma_1) \quad\quad  
x_0  \dot\sigma_2:=x_0\cdot X_{\widehat H}(\sigma_2)=\gamma(\sigma_2)/\sigma_1\]
i.e. Definition \ref{orbredudef} applies with $\mu=x_0 $.
The orbitally equivalent system 
\[
\dot x=H(x)=\frac1{1+\nu g_0(x)}\widehat H(x)=\left(\begin{array}{c} 1\\
  h_1(x)  
\\
  h_2(x)
\end{array}\right):=\left(\begin{array}{c} 1\\
                                       \frac {x_2+\nu x_1}  
{1+\nu x_0}\\
          \frac{\gamma(x_2)}{x_1(1+\nu x_0)}
\end{array}\right)
\]
remains orbitally reducible by the invariants of $g$. The 
coordinate change is now
\[
\begin{array}{rcl} y_0&=& x_0\\
                              y_1&=& x_1\\
                  y_2&=& h_1(x)=\frac{x_2+\nu x_1} {1+\nu x_0}.
\end{array}
\]

\medskip\noindent (i)
In the particular case of constant $\nu$, the system in new coordinates is
\begin{equation}\label{ex1nc}
\dot y= H^*(y)=\left(\begin{array}{c} 1\\
                                       y_2           \\
                                      p^*(y)
\end{array}\right) :=\left(\begin{array}{c} 1\\
                                       y_2           \\
                                      \frac 
{\gamma(x_2)}{y_1(1+\nu y_0)^2} 
\end{array}\right) 
\end{equation}
where 
\[x_2=y_2(1+\nu y_0)-\nu y_1.\]
This system is orbitally reducible by the invariants of $g$, expressed in new 
coordinates, and we note that $\lambda=\nu/(1+\nu)$ according to Remark \ref{lamfind}. The invariants are 
\[\widetilde\sigma_1=y_1/y_0 \quad\quad \widetilde\sigma_2=y_2(1+\nu y_0)-\nu y_1 
\]
and the orbital reduction is obtained with $\widetilde\mu=y_0(1+\nu y_0)$. 
In detail
\begin{equation}\label{ex1eq}
\begin{array}{rcl} \widetilde\mu\cdot X_{H^*}(\widetilde\sigma_1)&=& (\widetilde\sigma_2-\widetilde\sigma_1)\\
                             \widetilde \mu\cdot X_{H^*}(\widetilde\sigma_2)&=& \gamma(\widetilde\sigma_2)/\widetilde\sigma_1 
\end{array}
\end{equation}
Choose now for instance $\gamma=x_2^2$. Then Equation \eqref{ex1} gives the ODE $x\ddot x=\dot x^2$ which 
is easily solvable, whereas the ODE resulting from \eqref{ex1nc} (with $x_0=y_0=t$ and $y_1=y$) is
\begin{equation}\label{ex1last} y\ddot y(1+\nu t)^2=(\dot y+\nu t\dot y-\nu
y)^2
\end{equation}
which seems to be not solvable by standard methods. However, from \eqref{ex1eq}
with $\gamma=\widetilde\sigma_2^2$ one deduces 
\[ \frac {\widetilde\sigma_2}{ \widetilde\sigma_1}-\log \widetilde\sigma_2={\rm const}\]
which expresses a first integral for system \eqref{ex1nc}. On the other hand,
a first integral for this system corresponds to a first integral for
the  resulting second order ODE \eqref{ex1last}. Indeed, it can be checked that
\[\frac t y (\dot y+\nu t\dot y-\nu y)-\log(\dot y+\nu t\dot 
y-\nu y)={\rm const}\]
is satisfied.

\medskip\noindent (ii) In the particular case that $\nu=x_2$ one has
\[
y_2=x_2\left(\frac{1+x_1}{1+x_2x_0}\right),\quad
x_2=\frac{y_2}{1+y_1-y_0y_2}=:\phi(y),\quad\lambda=\frac{y_2}{1+y_1+y_2(1-y_0)}.
\]
The invariants, expressed in new coordinates, are now $y_1/y_0$ and $\phi(y)$.
From the third entry of $H^*$ one sees that 
\[
\dot y_2=\frac{\gamma(\phi)\cdot y_2}{y_1\cdot\phi}+\phi\cdot y_2-\phi\cdot\frac{y_1\cdot\phi + \gamma(\phi)\cdot y_0}{(1+\phi\cdot y_0)^2}
\]
As a particular example with $\gamma=0$ we obtain the second-order equation
\[
\ddot y=\left(\frac{1}{1-y-t\dot y}-\frac{y}{(1+y)^2}\right)\dot y^2
\]
which is reducible to a first order equation via the invariants of $g$.

\medskip
\noindent{\it Example 2.} The previous example still remains in the classical setting of lambda symmetries, since we started with a point symmetry. For an example in a more general setting, start with the simple system
\begin{equation}\label{ex12}
\begin{array}{lll}
\dot x_0&=&1\\
\dot x_1&=& x_2 \\ 
\dot x_2&=& 0
\end{array}
\end{equation}
which corresponds to $\ddot x=0$ and admits the infinitesimal symmetry
\[
g=\left(\begin{array}{c}  x_2\\
                                        x_1\\
                                       x_2
\end{array}\right)
\]
Given any scalar function $\nu$ the equation
\[
\dot x=H(x):=\frac{1}{1+\nu x_2}\left(\begin{array}{c} 1+\nu x_2\\
                                       x_2+\nu x_1\\
                                      \nu x_2
\end{array}\right)
\]
is orbitally reducible by the invariants of $g$. An independent set of invariants of $g$ is given by 
\[ \sigma_1=x_0 - x_2 \quad\quad  \sigma_2=x_2/x_1.\]
Let us consider the special case with $\nu =x_1$. The coordinate transformation then is given by 
\[
y_0=x_0,\quad y_1=x_1,\quad y_2=\frac{x_1^2+x_2}{1+x_1x_2},\quad x_2=\frac{y_2-y_1^2}{1-y_1y_2}
\]
and we have
\[
\lambda=\frac{y_1(1-y_1y_2)}{1-y_1^3}.
\]
In the $x$-coordinates we find
\[
\begin{array}{rcl}
X_H(\sigma_1)&=& 1/(1+x_1x_2)\\
X_H(\sigma_2)&=& -\sigma_2^2/(1+x_1x_2)
\end {array}
\]
and thus we have orbital reduction of  $\dot x=H(x)$ to the system
\begin{equation}\label{secordred}
\begin{array}{rcl}
\dot z_1&=& 1\\
\dot z_2&=& -z_2^2.
\end {array}
\end{equation}
We have a reduced (autonomous, in this particular case) one-dimensional equation $dz_2/dz_1=-z_2^2$.

Going to new coordinates $y$, the system $\dot x=H(x)$ is equivalent (by straightforward computation) to the second-order equation 
\[
\ddot y=\frac{(y\dot y -y^3)}{(1+y^3)^2}\left(1+y\dot y \right)^2+\frac{ \dot y-y^2+2y\dot y-\dot y^3-y^2\dot y^2}{1+y^3},
\]
which therefore is reducible to Equation \eqref{secordred} by the invariants $\widetilde \sigma_1$ and  $\widetilde \sigma_2$ (expressed in new coordinates). Note that $\dot x=H(x)$ is autonomous and therefore admits time translation as a symmetry, which is reflected in the one-dimensional orbitally reduced equation also being autonomous. (We chose this system for the sake of brevity, but note that the reduction to \eqref{secordred} is not a symmetry reduction.)
\subsection{An example of order three}

The essential point of this example is to illustrate nontrivial "joint lambda" symmetries. We start with the simple system
\begin{equation}\label{ex32}
\begin{array}{lll}
\dot x_0&=&1\\
\dot x_1&=& x_2 \\ 
\dot x_2&=& x_3\\ 
\dot x_3&=& 0
\end{array}
\end{equation}
corresponding to the third-order equation $x^{(3)}=0$.
This system admits the infinitesimal orbital symmetries
\[
g^{(1)}=\left(\begin{array}{c}  x_0\\
                                        x_1\\
                                       0\\
                                      -x_3
\end{array}\right),\,
g^{(2)}=\left(\begin{array}{c}  0\\
                                        x_1\\
                                       x_2\\
                                       x_3
\end{array}\right);\,\left[g^{(1)},\,g^{(2)}\right]=0,\, \left[g^{(1)},\,f\right]=f,\,\left[g^{(2)},\,f\right]=0.
\]
Note that $g^{(1)}$  and $g^{(2)}$ both are point symmetries.
A set of independent common invariants of $g^{(1)}$  and $g^{(2)}$ is given by
\[
\psi_1=x_0x_2/x_1,\, \psi_2=x_0^2x_3/x_1.
\]
We construct an equation that is reducible by the common invariants of $g^{(1)}$  and $g^{(2)}$. Let $\nu^{(1)}(x)=x_1$ and $\nu^{(2)}(x)=1/x_1$, thus
\[
\widehat H(x)=\left(\begin{array}{c}1+x_0x_1\\
                                                    x_2+x_1^2+ 1\\
                                                     x_3+x_2/x_1\\
                                                     -x_1x_3+x_3/x_1\end{array}
                                      \right)
\]
and
\[
 H(x)=\left(\begin{array}{c}1\\
                                                 \frac{   x_2+x_1^2+1}{1+x_0x_1}\\
                                                    \frac{ x_1x_3+x_2}{x_1(1+x_0x_1)}\\
                                                   \frac{  (1-x_1^2)x_3}{x_1(1+x_0x_1)}\end{array}
                                      \right)
\]
Let us look at the reduction first. One finds
\[
\begin{array}{rcl}
X_{\widehat H}(\psi_1)&=&\frac{1}{x_0}\left(\psi_1-\psi_1^2-\psi_2\right)\\
X_{\widehat H}(\psi_2)&=&\frac{1}{x_0}\left(2\psi_2-\psi_1\psi_2\right)\\
\end{array}
\]
and therefore $\psi_1$ and $\psi_2$ provide an orbital reduction of $\widehat H$ (as well as of $H$) to the autonomous two-dimensional system
\begin{equation}\label{jointlamred}
\begin{array}{rcl}
\dot z_1&=& z_1-z_1^2+z_2\\
\dot z_2&=& 2z_2-z_1z_2,
\end{array}
\end{equation}
which may be rewritten as a non-autonomous first order equation:
\[
\frac{dz_2}{dz_1}= \frac{2z_2-z_1z_2}{z_1-z_1^2-2z_2}
\]
The example was primarily chosen to obtain an explicitly invertible coordinate transformation towards the third-order equation. Since some expressions are somewhat unwieldy, we will write them down only in an abbreviated version. Passing to new coordinates, we set 
\[
y_0=x_0,\,y_1=x_1,\, y_2=X_H(y_1)=\frac{1+y_1^2}{1+y_0y_1}+\frac{x_2}{1+y_0y_1}
\]
which yields
\[
x_2=y_2(1+y_0y_1)-(1+y_1^2).
\]
Using the "hybrid" expressions
\[
X_H(x_2)=\frac{x_3}{1+y_0y_1}+\frac{x_2}{y_1(1+y_0y_1)}, \quad X_H(x_3)=\frac{1-y_1^2}{y_1(1+y_0y_1)}x_3,
\]
routine calculations provide
\[
y_3=\frac{-y_1^2-y_0y_1y_2+2y_1^2y_2-y_1^4+y_0y_1^3y_2+(1-y_1^2-y_0y_1y_2)x_2+y_1x_3}{y_1(1+y_0y_1)^2},
\]
which (being linear in $x_3$) can easily be solved for $x_3$ as a function of $y_0,\ldots,y_3$. Taking the Lie derivative $X_H(y_3)$ and
making the usual identifications, one obtains a (lenghty) third order equation for $y=y_1$, which can be reduced to \eqref{jointlamred} by the invariants $\tilde \psi_1$ and $\tilde\psi_2$ (expressed in the $y$-coordinates).

The example shows that the construction of nontrivial reducible higher order equations via prescribed "joint lambda" symmetries is feasible. But it also illustrates that work remains to be done towards a systematic approach.



\begin{thebibliography}{99}

\bibitem{BtD} Br\"ocker, T., and T. tom Dieck, 
\newblock ``Representations of compact Lie groups,"
\newblock Springer, New York - Berlin 1985.

\bibitem{CGWsprol} Cicogna, G., and G. Gaeta, and S. Walcher, 
\newblock {\em A generalization of $\lambda$-symmetry reduction for systems of ODEs: $\sigma$-symmetries},
\newblock J. Phys. A: Math. Theor. {\bf 45} (2012),  355205 (29 pp.).

\bibitem{HaWa} Hadeler, K.P., and S. Walcher, 
\newblock {\em Reducible ordinary differential equations},
\newblock  J.~Nonlinear Sci. {\bf 16} (2006), 583 - 613.

\bibitem{Her} Hermann, R.,
\newblock ``Differential geometry and the calculus of variations,"
\newblock Academic Press, New York 1968.

\bibitem{KraVi } Krasil'shchik, I.S., and A.M. Vinogradov (Eds.),
\newblock ``Symmetries and Conservation Laws for Differential Equations of Mathematical Physics,"
\newblock Translations of Mathematical Monographs, AMS, Providence 1999.

\bibitem{Mor} Morando, P.,
\newblock {\em Deformation of Lie derivative and $\mu$-symmetries},
\newblock J. Phys. A {\bf 40} (2007), 11547 - 11559.

\bibitem{MuRo1} Muriel, C., and J.L. Romero, 
\newblock {\em New methods of reduction for ordinary differential equations},
\newblock IMA J. Appl. Math. {\bf 66} (2001), 111 - 125.

\bibitem{MuRo1a} ---,
\newblock {\em $C^\infty$-symmetries and integrability of ordinary differential equations},
\newblock Proceedings of the I Colloquium on Lie theory and applications (2002), 143 - 150.

\bibitem{MuRo2} ---,
\newblock {\em $C^\infty$-symmetries and reduction of equations without Lie point symmetries},
\newblock J.~Lie Theory {\bf 13} (2003), 167 - 188.

\bibitem{MuRo3} ---,
\newblock {\em First integrals, integrating factors and $\lambda$-symmetries of second-order differential equations},
\newblock  J. Phys. A: Math. Gen. {\bf 42} (2009), 365207 (17 pp.).

\bibitem{NuLea} Nucci, M.C., and P.G.L. Leach, 
\newblock {\em The determination of nonlocal symmetries by the technique of reduction of order},
\newblock J.~Math. Anal. Appl. {\bf 251} (2000), 878 - 884.

\bibitem{Olv} Olver, P.J.,
\newblock ``Applications of Lie groups to differential equations,"
\newblock Springer, New York 1986.

\bibitem{OlRo2} Olver, P.J., and P. Rosenau, 
\newblock {\em Group-invariant solutions of differential equations},
\newblock SIAM J. Appl. Math. {\bf 47} (1987), 263 - 278.

\bibitem{PuSa} Pucci, E., and G. Saccomandi, 
\newblock {\em On the reduction methods for ordinary differential equations},
\newblock J. Phys. A: Math. Gen. {\bf 35} (2002), 6145 - 6155.

\bibitem{Ste} Stephani, H.,
\newblock ``Differential equations: Their solution using symmetries,"
\newblock Cambridge Universiy Press, Cambridge 1989.

\bibitem{WMul}  Walcher, S.,
\newblock {\it Multi-parameter symmetries of first order ordinary differential equations},
\newblock  J. Lie Theory {\bf 9} (1999), 249 - 269.

\end{thebibliography}
\end{document}